\documentclass[aps,prd,nofootinbib,superscriptaddress,notitlepage]{revtex4-1}

\usepackage{amsmath,amsthm,amssymb,enumerate}

\usepackage{multirow}
\usepackage{mathtools}
\usepackage{wasysym}
\usepackage{enumitem}
\usepackage{tensor}
\usepackage{booktabs}

\usepackage{etoolbox}

\usepackage[usenames]{color}
\usepackage{bm}

\newcommand{\tho}{{\textrm{\thorn}}}

\def \BEA { \begin{eqnarray}}
\def \EEA {\end{eqnarray}}
\def \BE {\begin{equation}}
\def \EE {\end{equation}}
\def\dd{\mathrm{d}}

\newcommand{\bg}{\mbox{\boldmath$g$}}

\def \WDS #1 {\mbox{$\Phi_{#1}^{S}$}}
\def \WDA #1 {\mbox{$\Phi_{#1}^{A}$}}
\def \WD #1 {\mbox{$\Phi_{#1}$}}

\def \mi {\stackrel{i}{m}}
\def \mj {\stackrel{j}{m}}
\def \mk {\stackrel{k}{m}}
\def \mr {\stackrel{r}{m}}
\def \ms {\stackrel{s}{m}}	
\def \mz {\stackrel{z}{m}}
\def \mq {\stackrel{q}{m}}
\def \mo {\stackrel{o}{m}}
\def \mD {\stackrel{2}{m}}
\def \mT {\stackrel{3}{m}}
\def \mC {\stackrel{4}{m}}

\newcommand{\sss}{\operatorname{s}}

\def \mio #1 {\mi_{#1}\ ^{  \! \! \! \! 0}} 
\def \mjo #1 {\mj_{#1}\ ^{  \! \! \! \! 0}} 
\def \mko #1 {\mk_{#1}\ ^{  \! \! \! \! 0}} 
\def \mro #1 {\mr_{#1}\ ^{  \! \! \! \! 0}} 
\def \mso #1 {\ms_{#1}\ ^{  \! \! \! \! 0}} 
\def \mpo #1 {\mp_{#1}\ ^{  \! \! \! \! 0}} 
\def \mzo #1 {\mz_{#1}\ ^{  \! \! \! \! 0}} 
\def \mqo #1 {\mq_{#1}\ ^{  \! \! \! \! 0}} 
\def \moo #1 {\mo_{#1}\ ^{  \! \! \! \! 0}} 
\def \mDo #1 {\mD_{#1}\ ^{  \! \! \! \! 0}} 
\def \mTo #1 {\mT_{#1}\ ^{  \! \! \! \! 0}} 
\def \mCo #1 {\mC_{#1}\ ^{  \! \! \! \! 0}} 

\def \miJ #1 {\mi_{#1}\ ^{  \! \! \! \! (1)}} 
\def \mjJ #1 {\mj_{#1}\ ^{  \! \! \! \! (1)}} 
\def \mkJ #1 {\mk_{#1}\ ^{  \! \! \! \! (1)}} 
\def \mrJ #1 {\mr_{#1}\ ^{  \! \! \! \! (1)}}

\def \bl {\mbox{\boldmath{$\ell$}}}
\def \hbl {\mbox{\boldmath{$\hat \ell$}}}
\def \bn {\mbox{\boldmath{$n$}}}

\def \hbn {\mbox{\boldmath{$\hat n$}}}

\newcommand{\be}{\begin{equation}}
\newcommand{\ee}{\end{equation}}
\newcommand{\beqn}{\begin{eqnarray}}
\newcommand{\eeqn}{\end{eqnarray}}

\newcommand{\ba}{\begin{array}}
\newcommand{\ea}{\end{array}}

\def \BEAH {\begin{eqnarray*}}
\def \EEAH {\end{eqnarray*}}
\def \BEA {\begin{eqnarray}}
\def \EEA {\end{eqnarray}}
\def \BDM {\begin{displaymath}}
\def \EDM {\end{displaymath}}

\def \T {\bigtriangleup  }

\newcommand{\tb}{\textcolor{black}}

\newtheorem{proposition}{Proposition}

\newtheorem{definition}[proposition]{Definition}
\newtheorem{remark}{Remark}

\newtheorem{lemma}{Lemma}

\def \H {\mathcal{H}}
\def \D {\mathrm{D}}

\newcommand*\Bg{\ensuremath{\boldsymbol{g}}}
\newcommand*\BR{\ensuremath{\boldsymbol{R}}}
\newcommand*\BC{\ensuremath{\boldsymbol{C}}}

\newcommand*\BS{\ensuremath{\textbf{S}}}

\newcommand*\BT{\ensuremath{\boldsymbol{T}}}

\begin{document}

\title{Almost universal spacetimes in higher-order gravity theories}

\date{\today}

\author{M. Kuchynka}
\email{kuchynkm@gmail.com}
\affiliation{Institute of Mathematics of the Czech Academy of Sciences, \v Zitn\' a 25, 115 67 Prague 1, Czech Republic}
\affiliation{Institute of Theoretical Physics, Faculty of Mathematics and Physics, Charles University in Prague,
V Hole\v{s}ovi\v{c}k\'ach 2, 180 00 Prague 8, Czech Republic}
\author{T. M\' alek}
\email{malek@math.cas.cz}
\author{V. Pravda}
\email{pravda@math.cas.cz}
\author{A. Pravdov\' a}
\email{pravdova@math.cas.cz}
\affiliation{Institute of Mathematics of the Czech Academy of Sciences, \v Zitn\' a 25, 115 67 Prague 1, Czech Republic}

\begin{abstract}
We study almost universal spacetimes -- spacetimes for which the field equations of any generalized gravity with the Lagrangian constructed from the metric, the Riemann tensor and its covariant derivatives of arbitrary order reduce to one single differential equation and one algebraic condition for the Ricci scalar. We prove that all $d$-dimensional Kundt spacetimes of Weyl type III   and traceless Ricci type N  are almost universal. Explicit examples of Weyl type II almost universal Kundt metrics are also given.  The considerable simplification of the field equations of higher-order gravity theories for almost universal spacetimes is then employed to study new Weyl type II, III, and N vacuum solutions to quadratic gravity in arbitrary dimension and six-dimensional conformal gravity. Necessary conditions for almost universal metrics are also studied.
\end{abstract}

\maketitle

\section{Introduction}

In effective field theories, the Einstein equations are modified by adding further terms to the Einstein--Hilbert action, leading to  the Lagrangian of the form
\be
L= G^{-1} (R-2 \Lambda)+f(\Bg,\BR,\nabla \BR,\dots)\,. \label{lagr}
\ee
The resulting field equations are, in most cases, considerably more complicated than the Einstein equations, and therefore very few exact solutions to modified gravities are known. Nevertheless, there exists a class of spacetimes, the so-called universal spacetimes \cite{Coleyetal08,HerPraPra14,Herviketal15}, for which all but one vacuum field equations of any theory of the form \eqref{lagr} are identically satisfied. The remaining field equation reduces to an algebraic constraint $\Lambda=F(I_i,\alpha_i)$ relating a cosmological constant $\Lambda$ with constant curvature invariants $I_i$  and constant parameters $\alpha_i$ of the theory. Thus, with an appropriate choice of  $\Lambda$, universal spacetimes are exact vacuum solutions to any theory of the form \eqref{lagr}.

Recently, it has been shown that for certain non-Einstein (and thus non-universal) spacetimes, field equations of any theory of the form \eqref{lagr} are also dramatically simplified. This has been observed in the case of  AdS waves and pp-waves \cite{gurses2013,Gurses:2014soa} and for Kerr-Schild-Kundt metrics with an (A)dS background \cite{Gurses:2016moi}. A closely related result in string theory, showing that gravitational waves in AdS do not receive any $\alpha'$ corrections, has been obtained in \cite{HorItz99}. All these spacetimes are of Weyl type N in the algebraic classification of tensors \cite{Coleyetal04} (see also \cite{OrtPraPra13rev} for a recent review).

In this paper, we set out to investigate these ``almost universal'' spacetimes in a much more general context. By studying necessary and sufficient conditions for ``almost universality'', we arrive at examples of Weyl type II, III, and N almost universal spacetimes. To make a connection with the Kerr-Schild approach of \cite{Gurses:2016moi}, we  also study almost universal Kerr-Schild spacetimes; however, instead of  (A)dS, we are able to use any type II, III, and N universal Kundt spacetime as a background metric. These results are also employed to construct new vacuum solutions to quadratic and cubic theories of gravity.

Let us proceed by introducing two classes of almost universal spacetimes, TN and TNS spacetimes (TNS $\subset$ TN).

\begin{definition}[\tb{Almost universal spacetimes}]
	\tb{Almost universal spacetimes (or equivalently TN spacetimes\footnote{TN(traceless type N) - all rank-2 tensors  constructed from the Riemann tensor and  its covariant derivatives of an arbitrary order are of traceless type N (i.e., of the form \eqref{TN}).})  are spacetimes, for which there exists a null vector $\bl$ such that for every symmetric rank-2 tensor $E_{ab}$} \tb{constructed polynomially from a metric, the Riemann tensor and  its covariant derivatives of an arbitrary order there exist a constant $\lambda$ and a function $\phi$ such that }
	\be
	E_{ab} =\lambda g_{ab}+\phi \ell_a \ell_b\,.\label{TN}
	\ee
	\tb{A TN spacetime is called TNS if in addition
		the last term in \eqref{TN} reduces to}
	\be 
	\phi\ell_a\ell_b=
	 \sum_{n=0}^{N}a_n\Box^n S_{ab}
	\,, \label{TNS}
	\ee
	\tb{where $S_{ab}$ is the traceless Ricci tensor and $a_i$ are constants. }
\end{definition}

\tb{Note that for TN spacetimes, tracelessness of $S_{ab}$ guarantees $\Box^n S_{ab} \propto \ell_a\ell_b $ and for TNS spacetimes, there are no terms involving the Weyl tensor  present in \eqref{TNS}, cf. \eqref{FKWC-C9}.}

It follows directly from the definition of TN spacetimes that all but two vacuum field equations of any theory of the form \eqref{lagr} hold identically. Furthermore, \tb{one of these two equations}, the equation corresponding to the $\lambda$ term in \eqref{TN}, 
reduces, similarly as in the case of universal spacetimes, to an algebraic equation
$\Lambda=F(I_i,\alpha_i)$ \tb{(see, e.g., eqs. \eqref{QG1} and \eqref{QG3})}.  Thus, the vacuum field equations lead only to one differential equation corresponding to the $\phi$ term in \eqref{TN}. Obviously, if the null radiation term is allowed, this equation can be omitted. 
  
\tb{Let us now briefly summarize selected  results of this paper.}

The main result of section \ref{sec_Nec} that focuses on necessary conditions for TN spacetimes reads 
\begin{proposition}[Necessary conditions for non-Einstein TN spacetimes]
	\label{nec_TN}
	Non-Einstein {TN} 
	spacetimes are necessarily CSI Kundt spacetimes
	of Weyl type II or more special.
\end{proposition}
This proposition holds in an arbitrary dimension. In the case of four dimensions, we arrive at a more general result including also Einstein TN spacetimes.

In section \ref{sec_sufficient}, we focus on sufficient conditions for  TN and TNS spacetimes, proving the 
following main results
\begin{proposition}[Sufficient conditions for TN spacetimes]\label{U}
	All $d$-dimensional Kundt spacetimes of Weyl type III \tb{or N} and traceless Ricci type N\footnote{Throughout this paper, by traceless Ricci type N we mean a Ricci tensor of the form \eqref{TN}.} are TN. 
\end{proposition}
We also show that certain classes of Weyl type II Kundt  spacetimes are TN.  

To study sufficient conditions  for TNS spacetimes,   we first generalize previous results of \cite{HerPraPra14,HerPraPra17} for type III universal spacetimes

\begin{proposition}[Necessary and sufficient conditions for Weyl type III universal Kundt spacetimes]\label{prop_U_III}
\tb{A Weyl type III  Kundt spacetime is universal if and only if 
it is Einstein and
\be
F_{0} \equiv {C^a}_{cde}C^{bcde}=0, \qquad  F_2\equiv {C^{pqrs}}_{;a}{C_{pqrs;b}}=0 .
\ee
}
\end{proposition}
For type III TNS spacetimes we then obtain
\begin{proposition}[Sufficient conditions for Weyl type III TNS spacetimes]\label{prop_TNS}
	Kundt spacetimes of Weyl type III  and traceless Ricci type N obeying $F_0 =0 = F_2$ are TNS.
\end{proposition}

Furthermore, in section \ref{sec_sufficient} we also study the Kerr--Schild transformation of  universal Kundt background spacetimes and show that under appropriate additional conditions,  the resulting spacetime is TN. 

Section \ref{sec_app} illustrates how 
\tb{vacuum field equations simplify for TN spacetimes of Weyl types II, III, and N for }
 specific higher-order gravities, such as quadratic and six-dimensional conformal gravities.

In the Appendix \ref{app_TIII}, we extend some of the results obtained in sections \ref{sec_Nec} and \ref{sec_sufficient} for TN spacetimes to a more general class of T-III spacetimes for which also most of the vacuum field equations of any theory of the form \eqref{lagr} are identically satisfied. Thus this class of spacetimes may  be also useful for constructing vacuum solutions of generalized gravities.
Appendix \ref{sec_KS} contains technical results on the Kerr--Schild transformation of the Einstein Kundt spacetimes employed in section \ref{subsec_KS}. 
In  Appendix \ref{FKWC}, we present the metric variations of all \tb{independent  curvature invariants in the 
Fulling--King--Wybourne--Cummins (FKWC) basis  \cite{FKWC}} up to order 6 for TN spacetimes of Weyl type III to provide examples of possible field equations for these spacetimes.
\tb{Finally, in Appendix \ref{sec_examples},
we extend section \ref{sec_app} by
providing explicit Weyl type III solutions to quadratic gravity constructed by the Kerr--Schild transformation of Weyl type III universal spacetimes.}

\section{Preliminaries}

In this paper, we employ the algebraic classification of tensors \cite{Coleyetal04}  (see \cite{OrtPraPra13rev} for review) and the \tb{higher-dimensional version of the Geroch--Held--Penrose} (GHP) formalism \cite{Durkeeetal10}.

We  work in a null frame in $d$ dimensions  consisting of null vectors $\bl$ and $\bn$ and $d-2$ spacelike vectors $\mbox{\boldmath{$m^{(i)}$}} $ obeying
\be
\ell^a \ell_a= n^a n_a = 0, \qquad   \ell^a n_a = 1, \qquad \ m^{(i)a}m^{(j)}_a=\delta_{ij},  \label{ortbasis}
\ee
where $a,b \in \{0,\dots,\ d-1\}$ and $i,j \in \{2,\dots,\ d-1\}$. Covariant derivatives along the frame vectors are denoted as
\be
D \equiv \ell^a \nabla_a, \qquad \T\equiv n^a \nabla_a, \qquad \delta_i \equiv m^{(i)a} \nabla_a .
\ee

\tb{For the frame vector $\bl$, we define
	Ricci rotation coefficients
	\BEA
	\kappa_i&=&\ell_{a;b}m_{(i)}^a\ell^b,\label{kappa}\\
	\tau_i&=&\ell_{a;b}m_{(i)}^a n^b,\label{tau}\\
\rho_{ij}&=&\ell_{a;b}m_{(i)}^a m_{(j)}^b,
\label{rho}
\EEA
where $\rho_{ij}$ is the so called optical matrix. The vector $\bl$ is geodetic iff $\kappa_i=0$.
Kundt spacetimes are defined as spacetimes for which there exists a geodetic null vector with vanishing optical matrix $\rho_{ij}$.}

A quantity $q$ has a {boost weight} (b.w.) ${\rm b}$ if it transforms according to
\be
\hat q = \lambda^{\rm b} q  
\ee 
 under boosts 
\be
\hbl = \lambda \bl, \qquad    \hbn = \lambda^{-1} \bn, \qquad  {\mbox{\boldmath{$\hat m^{(i)}$}}} = \mbox{\boldmath{$m^{(i)}$}}.  \label{boost}
\ee

Components of tensors in the  null frame  have distinct integer b.ws. Boost order of a tensor is the maximum b.w.\ of its non-vanishing components.
In an adapted null frame for algebraically special tensors, the highest b.w.\ tensor components vanish. In particular, the traceless Ricci tensor $S_{ab}$ in general admits components of b.w.\ $-2$,\dots,\ 2 while for  $S_{ab}$ of type N, only the b.w.\ $-2$ component $\omega'$ is non-trivial
\tb{\be
S_{ab}=\omega'\ell_a\ell_b.\label{Sab}
\ee}
 Similarly, in general the Weyl tensor admits components of b.w.\ $-2$,\dots,\ 2 while for a Weyl type II tensor,
only components of b.w.\ 0 ($\Phi$, $\WDA{ij} $, $\Phi_{ij}$, $\Phi_{ijkl}$), $-1$ ($\Psi'_i$, $\Psi'_{ijk}$) and $-2$ ($\Omega'_{ij}$) can be non-vanishing.
For Weyl type III/N, components of b.w.\ ($-1$, $-2$)/($-2$) can be non-vanishing, respectively (see \cite{OrtPraPra13rev}).    

Let us go back to the traceless Ricci tensor $S_{ab}$ of type N. While (in the adapted frame) $S_{ab}$ admits only the b.w.\ $-2$ component, in general its covariant derivatives possess also non-trivial components of b.w.\ $>-2$. However, for certain Kundt spacetimes, it can be shown that b.w.\ $>-2$ components vanish also for arbitrary covariant derivatives of  $S_{ab}$, $\nabla^{(I)} \BS$, i.e., boost order does not change with a covariant differentiation. Similarly, in these spacetimes,   covariant derivatives of the Weyl tensor do not increase boost order.   This is the key step for proving universality or the TN property. If one knows that e.g. $\nabla^{(I)} \BC $ admits only b.w.\ $\leq -2$ components then it immediately follows that all rank-2 tensors quadratic in  $\nabla^{(I)} \BC $ or of higher order vanish identically (rank-2 tensors admit components of b.w.\  $\geq -2$ only). This subsequently  leads to a considerable simplification of the analysis of non-trivial rank-2 tensors constructed from the curvature.

For proving that 
boost order of a tensor in certain Kundt spacetimes does not increase with a covariant differentiation (see section \ref{sec_sufficient}), we need to define the notion of balanced  tensors introduced in \cite{Pravdaetal02} 
and \cite{HerPraPra14}.

\begin{definition}[$k$-balanced tensors]
  In a frame parallelly  propagated along a null geodetic affinely parameterized vector field $\bl$, a tensor $\BT$ is said to be $k$-balanced, $k \in \mathbb{N}_{0}$,  if its boost weight $b$ part\footnote{i.e., $\BT_{(b)}$ is obtained by setting all components of $\BT$ with b.w.\ not equal to $b$ to zero.} $\BT_{(b)}$ satisfies $\BT_{(b)} = 0$ for $b \geq -k$ and 
$D^{-b-k} \BT_{(b)} = 0$ for $b<-k$. 
If $\BT$ is $0$-balanced, we say it is balanced.  
\end{definition}

A straightforward extension of Lemma A.7 of \cite{universalMaxwell} 
 to arbitrary $k$ then reads
\begin{lemma}\label{kbalancedderivative}
In a degenerate Kundt spacetime, a covariant derivative of a $k$-balanced tensor is again a $k$-balanced tensor. 
\end{lemma}

\section{Necessary conditions for almost universal spacetimes}
\label{sec_Nec}

In this brief section on necessary conditions for TN spacetimes,
we give a proof of proposition \ref{nec_TN} for non-Einstein spacetimes 
and discuss possible extensions to Einstein spacetimes in various special cases.

First, let us prove the following:
\begin{lemma}\label{lem_CSI}
	TN  spacetimes are CSI.
\end{lemma}
\begin{proof}
Let us assume that a spacetime possesses a non-constant curvature invariant $I$ constructed polynomially from the Riemann tensor and its covariant derivatives of arbitrary order.
$I$ can be expressed as a trace of a rank-2 tensor $E_{ab}$. Since
the trace of $E_{ab}$ is non-constant then $E_{(ab)}$ is not of the form
\eqref{TN} 
or more special and the spacetime is not TN.
 Thus TN$\subset$CSI.
\end{proof}
Now let us present a proof of proposition \ref{nec_TN}.

\vspace{1mm}
\noindent
{\it{\tb{Proof of proposition \ref{nec_TN}:}}}
			TN spacetimes are CSI thanks to lemma \ref{lem_CSI}.

		Geodeticity of $\bl$:
		         The contracted Bianchi identities $R^{ab}_{\  \ ;b}=0$ for the Ricci tensor of the form \eqref{TN} imply $ { \ell \indices{^a_{;b}} } \ell^b\propto \ell^a$ and thus $\bl$ is geodetic \tb{(i.e., $\kappa_i=0$ \eqref{kappa})}.

	$\bl$ is a  Kundt vector field:
		For the traceless part of the Ricci tensor, $S_{ab}$, of type N \tb{(i.e., of the form \eqref{Sab})}, the b.w.\ 0 component of $\Box S_{ab}$ reads 
		\be
		(\Box S_{ab}) \ell^a n^b = - \omega' \rho_{ij} \rho_{ij}\,,
		\ee
		\tb{where $\rho_{ij}$ is the optical matrix \eqref{rho}.}
		Tensor $\Box S_{ab}$ is  traceless and thus for {TN} spacetimes, all its b.w.\ 0 components have to vanish, which
        implies 
		\be
		\rho_{ij}=0
		\ee
		and therefore $\bl$ is a Kundt null congruence.

		TN spacetimes are of Weyl type II or more special:
		Now, taking into account the vanishing of $\kappa_i$ and  $\rho_{ij}$, the higher dimensional Newman--Penrose (NP) equations (NP1) and (NP3) of  \cite{Durkeeetal10} imply that the Weyl tensor components $\Omega_{ij}$ and $\Psi_{ijk}$ vanish and thus the Weyl tensor is of type II or more special.
\qed

Note that, from the proof of proposition \ref{nec_TN}, it follows that already TN$_2$ spacetimes\footnote{\tb{TN$_k$ spacetimes are defined similarly as TN spacetimes, however,  only derivatives of the Riemann tensor up to k$^{\rm th}$ order are considered.}} are necessarily Kundt and algebraically special. 

\begin{remark}[Necessary conditions for Einstein Weyl type N and III TN spacetimes]\label{rem_EinN}
\rm{So far, we have discussed {\it non-Einstein} TN  spacetimes.
It has been shown in lemma 4.8 of \cite{HerPraPra14} that Weyl type N CSI
Einstein spacetimes are Kundt.
Thus taking into account lemma  \ref{lem_CSI}, it follows that 
	Einstein {TN}  
	spacetimes of Weyl type N are necessarily CSI Kundt spacetimes.
	The same results hold also for Einstein TN spacetimes of Weyl type III
	under some additional genericity assumptions (see section 5.2 of \cite{HerPraPra14}).}
\end{remark}

\subsection{Necessary conditions in four and five dimensions}

In four dimensions, more general results can be obtained.

First, let us study necessary conditions for algebraically special TN  spacetimes.
In the four dimensional NP notation and in an appropriately chosen frame, the Weyl and Ricci tensors admit $\Psi_2$,  $\Psi_3$,  $\Psi_4$ and  $\Phi_{22}$ components, respectively.

For  {TN}$_0$, $\Psi_2$ is constant. The Bianchi equations (7.32a), (7.32b),
(7.32e) and (7.32h) from \cite{Stephanibook} then give
\be
\kappa=\sigma=\rho=\tau=0 \, ,
\ee
respectively, and thus

\begin{lemma}\label{lem_4D_TNnec}
	In four dimensions, genuine Weyl type II and D  {TN}$_0$  spacetimes are recurrent Kundt spacetimes.
\end{lemma}

Combining the results of proposition \ref{nec_TN}, remark \ref{rem_EinN}
(in four dimensions, the genericity assumptions for
type III always hold)
and lemma \ref{lem_4D_TNnec}, we arrive at a result that applies to both Einstein and non-Einstein spacetimes
\begin{proposition}[Necessary conditions for TN spacetimes in 4d]

In four dimensions, algebraically special TN spacetimes are necessarily Kundt.	
	
\end{proposition}

In contrast, in five dimensions,
	it can be shown using the same arguments as for the non-existence of genuine type II or D universal spacetimes in five dimensions (see section 4 of \cite{Herviketal15}) that
\begin{proposition}[Non-existence of 5d TN Weyl type II spacetimes]
	\label{prop_nonE5DTN}
	In five dimensions, genuine Weyl type II and D  {TN}$_0$  spacetimes do not exist.	
\end{proposition}

\section{Sufficient conditions for almost universal spacetimes}
\label{sec_sufficient}

In this section, we prove that various classes of Weyl types \tb{II, III, and N  Kundt spacetimes are TN and in some cases \tb{even} TNS. We also show that TN and TNS spacetimes can be constructed using an  appropriate Kerr-Schild transformation with a Weyl type II, III, or N  universal Kundt background.}

\subsection{Sufficiency for Weyl type III/N TN spacetimes}

First, let us prove that Kundt spacetimes of Weyl type III \tb{or N} and traceless Ricci type N  are TN (proposition \ref{U}).
We use a similar approach  as in the proof of Theorem 1.3 of \cite{HerPraPra14}.

\vspace{1mm}
\noindent
{\it{\tb{Proof of proposition \ref{U}:}}} 
	Note that by proposition 3.1 of \cite{KucPra16}, for Kundt spacetimes of Weyl type III/N and traceless Ricci type N, the Ricci and Weyl tensors are necessarily aligned.
	
	By proposition A.8 of \cite{universalMaxwell}, for 	Kundt spacetimes of Weyl type III/N and traceless Ricci type N, all covariant derivatives  of the Riemann tensor $\nabla^{(k)} \BR$ are of aligned type III (i.e., all non-vanishing components of  $\nabla^{(k)} \BR$ have negative boost weight).
	It immediately follows that all rank-2 tensors at least quadratic in  $\nabla^{(k)} \BR$
	(with $k\geq 0$) are of the form \eqref{TN}. Thus, it remains to show that the same result holds also for rank-2 tensors linear in  $\nabla^{(k)} \BR$.

	For the Weyl and Ricci type N Kundt spacetimes, it has been shown that
	$\nabla^{(k)} \BR$ ($k>0$) are of boost order at most $-2$ (see Proposition A.2 of \cite{EMclanek}) and thus all rank-2 tensors linear in  $\nabla^{(k)} \BR$
	are of the form \eqref{TN}. 
	
	Weyl type III case needs a more detailed discussion. 
	\tb{Note that  $S_{ab}$ of the form \eqref{Sab} is 1-balanced which follows from the primed version of eq. (2.50) of \cite{ghpclanek} (see also sec. 2.7 therein),}
	\tb{\be
    {\tho
     \omega' =D\omega' = 0,}
	\ee}
	 and therefore from lemma A.7 of \cite{universalMaxwell} (cf.\ also lemma \ref{kbalancedderivative}), $\nabla^{(k)} S_{ab}$ is 1-balanced as well. Thus, boost order of
	$\nabla^{(k)} S_{ab}$ is at most $-2$.
	
	In order for the contraction of $\nabla^{(k)} \BR$ to result in a rank-2 tensor, $k$ has to be even. 
	We will show that all rank-2 tensors linear in  $\nabla^{(k)} \BR$ are of the form \eqref{TN} using mathematical induction.

	First, consider the $k=2$ case. A change of the order of covariant derivatives in $\nabla^{(2)} \BR$,
	\begin{equation}\label{ricci}
	[\nabla,\nabla]\BR = \sum_\sigma \BR * \BR ,
	\end{equation}
	will result only in additional terms of b.w.\ $-2$ in contractions of RHS of \eqref{ricci}  due to the tracelessness of LHS of \eqref{ricci}.
	Employing the Bianchi identity (${R}_{ab[cd;e]}=0$) and changing the order of covariant derivatives if needed, any rank-2 tensor linear in 
	$\nabla^{(2)} \BR$ is either zero or of type N.

	Concerning $k > 2$, we proceed by induction. Assume that any change of the order of covariant derivatives in rank-2 tensors linear in $\nabla^{(n)} \BR$ produces only additional terms of b.w.\ $-2$ 
	and that any rank-2 contraction of $\nabla^{(n)} \BR$ is of type N. 
	Now, we prove that the same holds also for $\nabla^{(n+2)} \BR$.

	To prove the first property,  
	it is sufficient to prove that any rank-2 contraction of
	$\boldsymbol{Q} \equiv 
	( \nabla_{J} [\nabla,\nabla] \nabla_{I} \BR)$ is of boost order $-2$.
	Here, $I,J$ is a pair of arbitrary multi-indices satisfying $|I| + |J| = n$. Applying the Ricci identity on $\nabla_{I} \BR$ and employing the Leibniz rule, one obtains
		\begin{equation}
\nabla_{J} [\nabla,\nabla] \nabla_{I} \BR = \sum_\sigma \nabla_{J} \left( \BR * \nabla_{I} \BR  \right)
= \sum_\sigma \sum_{K \subset J} \binom{|J|}{|K|} \nabla_{K} \BR * \nabla_{J \setminus K}\nabla_{I} \BR.
\end{equation}
	Note that all terms with $|K|>0$ and $|I|^2+|J \setminus K|^2>0$ 
	are of boost order at most $(-2)$, while the $|K|=0$ terms and 
	$|I|^2+|J \setminus K|^2=0$ terms correspond to $\BR * \nabla^{(n)} \BR$.  Our induction assumption implies that rank-2 contractions of all these terms lead to a type N tensor and hence a rank-2 tensor linear in $\nabla^{(n+2)} \BR$ 
	is of type N. 
	Using this result together with the Bianchi identity and noting again that 
	$S_{ab}$ is 1-balanced,  the second property of $\nabla^{(n+2)} \BR$ follows. Therefore, both properties hold for all $\nabla^{(k)} \BR$, $k$ even. As a consequence, the terms linear in $\nabla^{(k)} \BR$ contribute 
	with b.w.\ $-2$ terms. 
\qed

\subsection{Sufficiency for Weyl type III/N TNS spacetimes}
\label{sec_sufTNS}

In the previous section, we have shown that Kundt spacetimes of Weyl type III \tb{or N} and traceless Ricci type N  are TN. Now, we  prove that some of these spacetimes are also TNS. While for Weyl type N, TN $\Leftrightarrow$ TNS
(cf. \cite{gurses2013}), for Weyl type III TN spacetimes, the Weyl tensor  and its derivatives in general contribute to b.w.\ $-2$ components of a  rank-2 tensor $E_{ab}$ (see appendix \ref{FKWC}). It turns out (cf., proposition \ref{prop_TNS}) that Weyl type N and III TNS spacetimes can be obtained from universal spacetimes by relaxing the Einstein condition and allowing  for type N traceless Ricci tensor. 

For this reason, we  generalize sufficient conditions for Weyl type III universal  spacetimes (proposition \ref{prop_U_III}).
It can be seen that the proof of universality of four-dimensional type III Einstein spacetimes with vanishing $F_2$
given in section 5.2 of \cite{HerPraPra17} can be straightforwardly generalized to arbitrary dimension, provided also $F_0=0$ (which is in four dimensions automatically satisfied) and thus Weyl type III Einstein Kundt spacetimes obeying $F_0=0=F_2$ are universal (proposition  \ref{prop_U_III}).
This generalizes the sufficient part of the Proposition 1.7 of \cite{HerPraPra17} to arbitrary dimension and theorem 1.4 of \cite{HerPraPra14} from the recurrent case ($\tau_i=0$) to a more general case $F_2=0$.
Indeed, it can bee seen from an expression for $F_2$ for Weyl type III Kundt spacetimes,
\be
F_2=-4\left\{\tau_i\tau_i(2\Psi'_j\Psi'_j-\Psi'_{jkl}\Psi'_{jkl})
-2\tau_i\tau_j\left[2\Psi'_{kil}\Psi'_{ljk}
-\Psi'_{ikl}\Psi'_{jkl}
+4\left(\Psi'_{k}\Psi'_{ijk}+\Psi'_{i}\Psi'_{j}\right)\right]\right\}\ell_a\ell_b\,,
\label{F2}
\ee
that $\tau_i=0$ implies $F_2 =0$, while $F_2 =0$ allows for non-vanishing $\tau_i$  (cf. \cite{HerPraPra17}).
\tb{Examples of Weyl type III universal spacetimes  with $\tau_i\not=0$ can be constructed
using a warp product
\be
 \dd \bar s^2 = \tfrac{1}{-\lambda z^2} ( \dd \tilde{s}^2 + \dd z^2),\ \ \lambda<0,
\label{warp_productU}
\ee
where $\dd \tilde{s}^2$ is a Ricci flat Weyl type III universal Kundt spacetime. Since the Weyl type \cite{OrtPraPra11} and the Kundt property (see eq. (B3) of \cite{OrtPraPra11}) are preserved under \eqref{warp_productU}
 and
 $\tilde F_0 = 0 = \tilde F_2$ implies $\bar F_0 = 0 = \bar F_2$ (see eq. (23)  of \cite{OrtPraPra11}),
the resulting spacetime is again a Weyl type III Einstein universal Kundt spacetime
with $\bar R = d(d-1)\lambda$:
\begin{proposition}\label{prop_warp_U}
	A warp product \eqref{warp_productU}, where  $\dd \tilde{s}^2$ is  a Ricci flat Weyl type III universal Kundt spacetime is a Weyl type III universal Kundt spacetime.
\end{proposition}
In general, the resulting spacetime has a non-vanishing $\tau_i$, i.e. it is not recurrent (c.f. also appendix \ref{sec_ex_Ricci}).}

	Let us proceed with proving that Weyl type III traceless Ricci type N  Kundt spacetimes obeying $F_0 =0 = F_2$ are TNS (proposition \ref{prop_TNS}).

\vspace{1mm}
\noindent
{\it{\tb{Proof of proposition \ref{prop_TNS}:}}} 
	By proposition A.8 of \cite{universalMaxwell}, 
	all tensors $\nabla^{(k)} \BR$ are of aligned type III.
	Furthermore, since $S_{ab}$ is 1-balanced (cf.\ proof of proposition \ref{U}) all derivatives of the Ricci tensor are of boost order $\leq -2$. 
	By counting boost weights, one can see that \tb{all mixed invariants (i.e., invariants constructed from both $\nabla^{(k)} \BS $  and $\nabla^{(l)} \BC$, $k,l\geq 0$) vanish and thus it is sufficient to consider rank-2 tensors $E_{ab}$ constructed purely from the Weyl tensor and its derivatives (at most quadratic in $\nabla^{(l)} \BC$) and purely from the   Ricci tensor and its derivatives (linear in $\nabla^{(k)} \BS$). }
	
\tb{	First, let us prove by the mathematical induction that all rank-2 tensors constructed from the Ricci tensor and its derivatives have the form \eqref{TN}, \eqref{TNS} (cf. \cite{gurses2013}).}

\tb{Obviously, the Ricci tensor of traceless type N has the form \eqref{TN}, \eqref{TNS}.
All rank-2 tensors constructed from  the second derivatives of the Ricci tensor ($\Box R_{ab}$,
$R^c_{\ a;cb}$, $R^c_{\ a;bc}$, $R^e_{\ e;ab}$)
either vanish or can be cast into the form \eqref{TN}, \eqref{TNS} using commutator 
	\be \label{commutator}
	[\nabla_c,\nabla_d]P_{e_1 e_2 \dots e_k}
	=-\sum_{i=1}^k {R^f}_{e_i cd} P_{e_1\dots e_{i-1} f e_{i+1} \dots e_k}
	\ee
	and the contracted Bianchi identity $\nabla_c R^c_a =
	\tfrac{1}{2}R_{;a}=0$.}

\tb{Now, assuming that all rank-2 tensors constructed from  the $k^{\rm th}$ derivative of the Ricci tensor have the form \eqref{TN}, \eqref{TNS}, we show
that this also holds for all rank-2 tensors constructed from  the $(k+2)^{\rm th}$ derivative of the Ricci tensor. 
Rank-2 tensors constructed from  the $(k+2)^{\rm th}$ derivative of the Ricci tensor can have free indices 
in three positions
\BEA
&1.& \ R_{\bullet\bullet;\dots}
 \label{SC1} \\
&2.& \ R_{\bullet .;\dots \bullet\dots} \label{SC2} \\
&3. &\ R_{. . ;\dots \bullet\dots\bullet\dots}
 \label{SC3}
\EEA
where free indices are indicated by $\bullet$. 
In the case 1., one can use commutator \eqref{commutator} to reshuffle covariant derivatives
and arrive at the form \eqref{TN}, \eqref{TNS}. The remaining part arising from the right hand side of commutator \eqref{commutator} contains only $k^{\rm th}$ derivatives of the Ricci tensor and thus have the form \eqref{TN}, \eqref{TNS} by our assumption. In the cases 2., 3.,
we can again reshuffle indices to obtain $(k+1)^{\rm th}$ derivative of the contracted Bianchi identity, i.e. $R^c_{\ \ .; c \dots}=0$, with
the remaining part arising from the right hand side of commutator \eqref{commutator} containing again only $k^{\rm th}$ derivatives of the Ricci tensor and thus having the form \eqref{TN}, \eqref{TNS} by our assumption.
}

		Now, let us proceed with terms constructed from the Weyl tensor and its derivatives.   	
		\tb{In 	 \cite{HerPraPra17}, it has been shown that for four-dimensional Weyl type III, Einstein Kundt spacetimes obeying $F_2=0$, all rank-2 
			tensors of the form $\nabla^{(k)}\BC\otimes\nabla^{(l)}\BC$  vanish. This proof can be straightforwardly generalized to the case of higher dimensions by adding an additional assumption $F_0=0$ which holds identically in four dimensions.}
		 \tb{Moreover, it is not affected by the presence of b.w. $-2$ terms in the Ricci tensor, thus it can be also straightforwardly generalized to the traceless Ricci type N case  and therefore terms $\nabla^{(l)}\BC \otimes \nabla^{(m)} \BC$ will not contribute to rank-2 tensors (cf., proposition 5.9 of \cite{HerPraPra17}).}	
Hence, only terms linear in $\nabla^{(k)}\BC$
		can contribute to a rank-2 tensor. Employing the  Bianchi
		\be
		\nabla^b C_{abcd} = \frac{d-3}{d-2} ( \nabla_d S_{ac} - \nabla_c S_{ad} ) \label{BianchiC}
		\ee
	and  Ricci identities, one arrives at 
	\be\label{Weylcontraction}
	  \nabla^d \nabla^b C_{abcd} = \frac{d-3}{d-2} \Box S_{ac} - \frac {(d-3)}{(d-2)(d-1)} R S_{ac},
	  \ee
	 being of the  TNS form \eqref{TN}, \eqref{TNS}.
	 
\tb{			  Using the mathematical induction, one can show that all rank-2 tensors linear in the Weyl tensor have the TNS form
	\eqref{TN}, \eqref{TNS}. We assume that all rank-2 tensors constructed from $k^{\rm th}$ derivative of the Weyl tensor have the form \eqref{TN}, \eqref{TNS}. Then using  commutator \eqref{commutator},  all rank-2 tensors constructed from $(k+2)^{\rm th}$ derivative of the Weyl tensor can be cast into the form $C^c_{\ \ \dots;c\dots}$,  which using \eqref{BianchiC}  and results from the previous paragraph  have the desired form, and additional terms containing  $k^{\rm th}$ derivative of the Weyl tensor that have the TNS  form \eqref{TN}, \eqref{TNS} by our starting assumption.} \qed

\subsection{Sufficiency for certain Weyl type II TN spacetimes}

Proposition \ref{U} addressing sufficient conditions for Weyl type III and N TN spacetimes does not include all TN spacetimes -- as we will show in this section
Weyl type II TN spacetimes also exist.

First, let us observe that one can construct  Weyl type II TN spacetimes from
Weyl type III and N TN spacetimes by taking a direct product with maximally symmetric spaces. A 
simple generalization  of proof of proposition 6.2  of \cite{Herviketal15} leads to a generalization of this
proposition  
to TN spacetimes
\begin{proposition}
	\label{prop_direct_TN}
	Let $M =  M_0 \times M_1 \times \dots \times M_{N-1}$, where $M_0$    is a  Lorentzian manifold and $M_1 \dots M_{N-1}$ are non-flat Riemannian maximally symmetric spaces. Let all blocks $ M_\alpha, \alpha=0 \dots N-1$, be of the same dimension and with the same value of the Ricci scalar $R_\alpha$.  If 
	$M_0$ is a  TN   Weyl type III or N  spacetime then $M$ is a type II TN spacetime.   
\end{proposition}

Now, we proceed with more general type II TN spacetimes.

\begin{remark}\label{rem_prop7.1}
	It is straightforward to generalize Proposition 7.1 of \cite{Herviketal15} for
	spacetimes with the Ricci tensor of the form \eqref{TN} leading to the following slightly more general result: for type II Kundt spacetimes with the Ricci tensor of the form \eqref{TN} admitting
	a  parallelly propagated null frame along an mWAND $\bl$, all covariant derivatives of the Riemann tensor $\nabla^{(k)}\BR $, $k\geq 1$, are at most of boost order $-2$ providing the following  three conditions hold:
	\begin{enumerate}
		\item 	b.w.\ $-1$ components of the Weyl tensor vanish,\\ \vspace{-6mm}
		\item for b.w.\ $-2$ components of the Weyl tensor, $\Omega'_{ij}$, $D\Omega'_{ij}=0$,\\\vspace{-6mm}
		\item the boost order of $\nabla^{1}C$ is at most $-2$. 
	\end{enumerate}
\end{remark}

Let us consider  higher-dimensional  generalizations of the Khlebnikov--Ghanam--Thompson metric \cite{Khleb,GT2001,Gibbons:2007zu} consisting of $N$ 2-blocks/3-blocks introduced in \cite{Herviketal15} in the context of universal spacetimes
\begin{equation}
\dd s^2 = 2 \dd u \dd v + ({\tilde \lambda} v^2 + H(u, x_\alpha, y_\alpha)) \dd u^2 + \frac{1}{|{\tilde \lambda} |} \sum_{\alpha=1}^{N-1} (\dd x_\alpha^2 + \sss^2(x_\alpha) \, \dd y_\alpha^2), \ \ \ 
\label{GT-2blocks}
\end{equation}
with $\sss(x_\alpha) = \sin (x_\alpha)$ for ${\tilde \lambda}  >0$,  $\sss(x_\alpha) = \sinh (x_\alpha)$ for ${\tilde \lambda}  <0$,
and
\begin{equation}
\dd s^2 = 2 \dd u \dd v + H(u, z, x_\alpha, y_\alpha, z_\alpha) \dd u^2 + 2 \frac{2v}{z} \dd u \dd z
- \frac{2}{{\tilde \lambda}  z^2} \dd z^2 - \frac{2}{{\tilde \lambda} } \sum_{\alpha=1}^{N-1} \left[\dd x_\alpha^2 + sh^2_\alpha \, (\dd y_\alpha^2 + s^2_\alpha \dd z_\alpha^2)\right],
\label{appGT3:GTmetric}
\end{equation}
with ${\tilde \lambda} <0$,   $s_\alpha = \sin (y_\alpha)$,  and $sh_\alpha = \sinh (x_\alpha)$, respectively.

Both these metrics obey all three conditions given in remark \ref{rem_prop7.1} and the Ricci tensor is of the form \eqref{TN} with $\phi=-\tfrac{1}{2}\Box H$
and $\phi\propto \Box H-2{\tilde \lambda}  zH_{,z}$, respectively.

By remark \ref{rem_prop7.1}, all non-vanishing rank-2 tensors constructed from $\nabla^{(k)}\BR$, $k\geq 0$, containing at least one term  with $k\geq 1$ are at most of boost order $-2$. Rank-2 tensors constructed from $\nabla^{(k)}\BR$, $k=0$,
are of the form \eqref{TN} (see sec. 7.2 of \cite{Herviketal15}).
We can conclude with 
\begin{proposition}
	Weyl type II metrics \eqref{GT-2blocks} and  \eqref{appGT3:GTmetric} are TN.
\end{proposition}

\subsection{Kerr-Schild transformations of universal spacetimes are TN }
\label{subsec_KS}

 To make a connection with the Kerr-Schild approach of \cite{HorItz99,Gurses:2016moi}, let us study TN and TNS spacetimes generated from universal spacetimes by the Kerr-Schild transformation. Results of this subsection in part overlap with those of subsection \ref{sec_sufTNS}. However, they apply also to Weyl type II. For some applications, the Kerr-Schild formulation of the results may be also more practical.

Let us consider spacetimes
\be
\Bg = \Bg_{\rm{UK}} +2{\cal H}\bl\otimes\bl\,,\ \ D{\cal H}=0\, , \label{KSUK}
\ee
where $\Bg_{\rm{UK}}$ are  universal Kundt spacetimes and $\bl$ is the Kundt null direction of the background spacetime $\Bg_{\rm{UK}}$. Note that $\Bg_{\rm{UK}}$ spacetimes are necessarily Einstein, algebraically special and degenerate Kundt. Thus they admit a metric of the form
\cite{ColHerPel09a,Coleyetal09}, 
\be
\dd s^2 =2\dd u\left[\dd r+H(u,r,x)\dd u+W_\alpha(u,r,x)\dd x^\alpha\right]+ g_{\alpha\beta}(u,x) \dd x^\alpha\dd x^\beta , \label{Kundt_deg}
\ee
where $\alpha,\beta=2 \dots n-1$ with
\be
W_{\alpha}(u,r,x)=rW_{\alpha}^{(1)}(u,x)+W_{\alpha}^{(0)}(u,x) , \ \ 	
H(u,r,x)=r^2H^{(2)}(u,x)+rH^{(1)}(u,x)+H^{(0)}(u,x) .
\label{Kundt_deg2}
\ee

Kundt null direction is $\bl=\dd u$ and thus the Kerr-Schild transformation  \eqref{KSUK} amounts to the transformation $H^{(0)}(u,x) \rightarrow H^{(0)}(u,x)+ {\cal H} $ in \eqref{Kundt_deg}, \eqref{Kundt_deg2}   and the resulting metric \eqref{KSUK} is clearly also Kundt degenerate metric.

From appendix \ref{sec_KS} it follows that the only changes in the curvature of the Kerr-Schild transformed metric \eqref{KSUK} appear in b.w.\ $-2$ components of the Weyl and Ricci tensors.

Since b.w.\ of ${\cal H}$ is $(-2)$ the Kerr--Schild perturbation $2{\cal H}\bl\otimes\bl$ in \eqref{KSUK}
is 1-balanced. Lemma A.7 of \cite{universalMaxwell} 
 then implies that
an arbitrary covariant derivative of $2{\cal H}\bl\otimes\bl$ remains 1-balanced. Thus
 arbitrary terms constructed from the perturbation 
and its derivatives entering an arbitrary rank-2 tensor $E_{ab}$ (constructed from the Riemann tensor of the full metric  \eqref{KSUK} and its derivatives)
influence only b.w.\ $(-2)$ components of $E_{ab}$  and thus do not violate the form \eqref{TN}.
Therefore, using also the results of appendix \ref{sec_KS}
we arrive at
\begin{proposition}
	For background metrics $\Bg_{\rm{UK}}$ of Weyl types II, III, and N, the Kerr-Schild transformation \eqref{KSUK} preserves the Weyl type and the resulting metric is TN.
  \label{prop_KS_Weyltypes}
\end{proposition}

For Kerr-Schild spacetimes with a flat or (A)dS background, the Kerr-Schild transformation \eqref{KSUK} gives a Weyl type N spacetime \cite{OrtPraPra09,MalPra11}. Thus clearly the Kerr-Schild transformation \eqref{KSUK} with universal backgrounds represents a more general class of spacetimes.

Since for Weyl type III Kundt spacetimes, the curvature polynomials $F_0 = (\tfrac{1}{2}\Psi'_{ijk} \Psi'_{ijk} - \Psi'_i \Psi'_i) \ell_a \ell_b$ 
 and $F_2$ \eqref{F2} are preserved by \eqref{KSUK} (see \eqref{KS-Omega}--\eqref{KS-Omegap} and \eqref{KS-tau}), from proposition \ref{prop_TNS} it follows
	\tb{\begin{proposition}
		The  Kerr-Schild transformation \eqref{KSUK} 
		of a   Weyl type III universal background metric $\Bg_{\rm{UK}}$ is a Weyl type III TNS spacetime.
		\label{prop_KS_TNS}
	\end{proposition}}

\section{Applications in higher-order gravities}
\label{sec_app}

Since  terms $\Box^n S_{ab}$ appear in most higher-order gravity field equations, cf. appendix \ref{FKWC}, we start with an examination of their form in TN spacetimes.

Let us assume a spacetime to be Kundt of
aligned Riemann type II
with the b.w.\ 0 part of the Ricci tensor  proportional to the metric and with
$\bl$ being the (affinely parametrized) Kundt vector.
Then, in a parallelly propagated frame with $\tau_{i} = L_{1i}$\footnote{Such a null frame always exists, cf. \cite{PraPra08}.}, for a function $f$ satisfying $D f = 0$, one has
\begin{align}\label{fll}
\Box(f \ell_a \ell_b) = \left[(\mathcal{D} - 2R_{0101}) f\right] \ell_a \ell_b,
\end{align}
where the differential operator $\mathcal{D}$ is defined as
\begin{equation}\mathcal{D} \equiv \Box  + 4\tau_{i}\delta_i  + 2\tau_{i}\tau_{i}  + 2 \Phi + 2\frac{d-2}{d-1}\frac{R}{d}.\label{calD}
\end{equation} 

Employing the Bianchi and Ricci identities and commutators \cite{Coleyetal04vsi}, one can see that for
TN spacetimes, 
the b.w.\ $-2$ components of $\Box^n S_{ab}$ are constant along geodesics generated by $\bl$ for any $n \in \mathbb{N}_0$. Using \eqref{fll}, 
one arrives at	
\begin{equation}\label{boxy}
\Box^n S_{ab} = \left[ (\mathcal{D} -2R_{0101})^n \omega^\prime \right] \ell_a \ell_b. 
\end{equation}

\subsection{Quadratic gravity}\label{QGsekce}
Let us apply the obtained results on the case of quadratic gravity. 
Its Lagrangian reads  
\begin{equation}
\mathcal{L}_{QG} = \textstyle{ \frac{1}{\kappa}}(R-2\Lambda_0) + \alpha R^2 + \beta R_{ab}R^{ab} + 
\gamma (R_{abcd}R^{abcd} - 4R_{ab} R^{ab} + R^2 ),
\end{equation}
where the cosmological constant $\Lambda_0$ and coupling constants $\kappa, \alpha, \beta, \gamma$ of the theory are fixed. 
The metric variation 
of the action corresponding to $\mathcal{L}_{QG}$ then yields \cite{GulluTekin2009}
\begin{align}\label{QGtensor}
E_{ab} =& {\textstyle{\frac{1}{\kappa}}}
\left( R_{ab} - {\textstyle{\frac{1}{2}}}
R g_{ab} + \Lambda_0 g_{ab} \right)
+ 2 \alpha R \left( R_{ab} -{\textstyle{\frac{1}{4}}}
R g_{ab} \right)
+ \left( 2 \alpha + \beta \right)\left( g_{ab} \Box - \nabla_a \nabla_b \right) R \nonumber \\
&+ 2 \gamma \bigg( R R_{ab} - 2 R_{acbd} R^{cd}
+ R_{acde} R_{b}^{\phantom{b}cde} - 2 R_{ac} R_{b}^{\phantom{b}c}
-{\textstyle{\frac{1}{4}}}
g_{ab} \left( R^{cdef}R_{cdef} - 4 R^{cd}R_{cd} + R^2 \right) \bigg) \nonumber \\
&+ \beta \Box \left( R_{ab} -{\textstyle{\frac{1}{2}}}
R g_{ab} \right)
+ 2 \beta \left( R_{acbd} -{\textstyle{\frac{1}{4}}}
g_{ab} R_{cd} \right) R^{cd}.
\end{align}
In this  section, two different classes of vacuum solutions within TN spacetimes will be discussed. 

\subsubsection{Weyl type III solutions}
For Weyl type III TN metrics satisfying $F_0 = 0$, the tensor $E_{ab}$ of \eqref{QGtensor} simplifies to \eqref{TNS} with 
\begin{align}
\lambda &=    \frac{\Lambda_0}{\kappa} - \frac{d-2}{2\kappa d}R - \frac{(d-4)(\alpha d + \beta)}{2d^2}R^2
- \frac{(d-4)(d-3)(d-2)}{2d^2(d-1)}\gamma R^2,\label{cg1} \\ 
a_0 &= \frac{1}{\kappa} + 2R \left( \alpha +\beta \frac{d-2}{d(d-1)} + \gamma \frac{(d-3)(d-4)}{d(d-1)} \right),\label{cg2} \\
a_1 &= \beta ,\label{cg3}
\end{align}
and with the rest of the coefficients $\{ a_i \}$ being zero. Hence, employing \eqref{boxy}, the vacuum field equations $E_{ab} = 0$ read 
\begin{equation}\label{QG1}
(d-4)   \frac{ (d-1)(\alpha d + \beta) + (d-3)(d-2) \gamma }{2d^2 (d-1)}  R^2
+ \frac{d-2}{2\kappa d}R = \frac{\Lambda_0}{\kappa},
\end{equation}
\begin{equation}\label{QG2}
\left\{ \mathcal{D} + \frac{1}{\kappa \beta} + 2R \left( \frac{1}{d} + \frac{\alpha}{\beta} + \frac{\gamma}{\beta}\frac{(d-3)(d-4)}{d(d-1)}  \right) \right\}\omega^\prime = 0.
\end{equation}
In the special case of a TN metric\footnote{\tb{Such a metric is necessarily TNS by proposition \ref{prop_KS_TNS}.}} obtained by the Kerr-Schild transformation \eqref{KSUK} of a universal Kundt metric, $\omega^\prime$ is given by $\omega^\prime = - \mathcal{D} \mathcal{H}$ \tb{(using \eqref{KS-omegap} with \eqref{Xi})}, and hence eq. \eqref{QG2} takes the form of a factorized 4th order differential equation whose solution reads $\mathcal{H} = \mathcal{H}_0 + \mathcal{H}_1$, where $ \mathcal{H}_0$, $\mathcal{H}_1$ are solutions to 2nd order equations
\begin{equation}
\mathcal{D} \mathcal{H}_0 = 0, \qquad \left\{ \mathcal{D} + \frac{1}{\kappa \beta} + 2R \left( \frac{1}{d} + \frac{\alpha}{\beta} + \frac{\gamma}{\beta}\frac{(d-3)(d-4)}{d(d-1)}  \right) \right\}\mathcal{H}_1 = 0,
\label{QGeq}
\end{equation}
respectively. \tb{Note that the Kerr-Schild transformation \eqref{KSUK}  with $\mathcal{H}=\mathcal{H}_0$ gives again an Einstein spacetime ($\omega'=0$)  while for $\mathcal{H}_1\not= 0$, it gives a non-Einstein solution to quadratic gravity (i.e., $\mathcal{H}_0$ only changes the Einstein background metric).}
\tb{While a type N subclass of these solutions was obtained already in \cite{MalekPravdaQG}, genuine Weyl type III solutions are new. Examples of such vacuum solutions to quadratic gravity are given 
 in Appendix \ref{sec_examples}.}

\subsubsection{Weyl type II solutions}
For the higher-dimensional generalization \eqref{GT-2blocks} of Khlebnikov-Ghanam-Thompson metrics consisting of $N$ 2-blocks, the tensor \eqref{QGtensor} reduces to 
\be
E_{ab} =  
\lambda^\text{(III)} g_{ab} + 
{a}_0^\text{(III)} S_{ab} +  
{a}_1 \Box S_{ab} + \left( 2\beta - 4\gamma \frac{d-4}{d-2} \right) S^{cd}C_{acbd}
+ 2\gamma C_{acde}C\indices{_b^{cde}} - \frac{1}{2}\gamma C_{cdef}C^{cdef}g_{ab},
\ee
with {$\lambda^\text{(III)}$, 
$a_0^\text{(III)}$, $a_1$} denoting the original coefficients \eqref{cg1}--\eqref{cg3}. On the first sight, $E_{ab}$ deviates from \eqref{TNS} by contributions from the Weyl tensor. However, due to a convenient structure of the b.w.\ 0 part of the Weyl tensor, one has
\begin{align}
&S^{cd}C_{acbd} = - \frac{d-2}{d-1}\frac{R}{d} S_{ab}, \label{KGT-SC}\\
& C_{acde}C\indices{_b^{cde}} = 2 \frac{R^2}{d^2}\frac{d-2}{d-1}g_{ab}, \label{KGT-CC}
\end{align}
so that $E_{ab}$ again reduces to the form \eqref{TNS}, but this time with coefficients 
\begin{align}
&\lambda =    \frac{\Lambda_0}{\kappa} - \frac{d-2}{2\kappa d}R - \frac{(d-4)(\alpha d + \beta + \gamma(d-2))}{2d^2}R^2 ,\label{qg1x} \\ 
& a_0 = \frac{1}{\kappa} + 2R\left( \alpha + \gamma \frac{d-4}{d} \right), \\
& a_1 = \beta,
\end{align}
and with the rest of coefficients $\{ a_i \}$ vanishing. Employing \eqref{boxy} again and noticing that $\mathcal{D}$ \eqref{calD} reduces to $\Box$, the  vacuum field equations of quadratic gravity read
\begin{equation}\label{QG3}
\frac{(d-4)(\alpha d + \beta + \gamma(d-2))}{2d^2}R^2  + \frac{d-2}{2\kappa d}R  = \frac{\Lambda_0}{\kappa},  
\end{equation}
\begin{equation}\label{QG4}
\left(  \Box + \frac{1}{\kappa \beta} + 2 R \frac{d \alpha + \beta + (d-4)\gamma}{ \beta d} \right) \Box H = 0 .
\end{equation}
While the solution $R$ of \eqref{QG3} determines the parameter ${\tilde\lambda} \equiv R/d$
in the KGT metric, the remaining metric function $H$ determined by the 4th order equation \eqref{QG4} can be  expressed as $H= H_0 + H_1$, where $H_0$ and $H_1$ solve 2nd order equations
\begin{equation}
\Box H_0 = 0 , \qquad \left(\Box + \frac{1}{\kappa \beta} + 2 R \frac{d \alpha + \beta + (d-4)\gamma}{d \beta}   \right) H_1 = 0,
\end{equation}
respectively.
To our knowledge, the solutions with $H_1\not= 0$ are the first known non-Einstein quadratic gravity vacuum solutions of Weyl type II in arbitrary even dimension.

\subsection{Conformal gravity in six dimensions}

Now, let us study vacuum solutions of conformal gravity in six dimensions given by the Lagrangian \cite{Metsaev2010,Lu2011}
\begin{equation}
\mathcal{L}_\text{conf} = \beta \left( R R_{ab}R^{ab} - \tfrac{3}{25} R^3 - 2 R^{ab}R^{cd}R_{acbd} - R^{ab}\Box R_{ab} + \tfrac{3}{10} R \Box R \right).
\label{CCG}
\end{equation}
The parameters of this theory are tuned in such a way that the field equations,
which can be found in the full form in \cite{Lu2011}, are satisfied by any metric conformal to an Einstein metric. 
Let us present vacuum solution of this theory that are not conformal to Einstein spacetimes.

\subsubsection{Weyl type III solutions}
For Weyl type III TN Kundt metrics, $E_{ab}$ of the theory $\mathcal{L}_\text{conf}$ takes the TNS form \eqref{TNS}
with non-vanishing coefficients
\begin{equation} 
a_0 = - \tfrac{4\beta}{75} R^2, \quad
a_1 = \tfrac{7\beta}{15} R, \quad
a_2 = -\beta,
\end{equation}
\tb{thus the conformal gravity field equations have the form}
\be
  \left(\Box -\tfrac{4R}{15}\right)\left(\Box -\tfrac{R}{5}\right)S_{ab}=0.
\ee
\tb{{For $R<0$, we obtain $(\Box +\frac{8}{\ell^2})(\Box + \frac{6}{\ell^2}) S_{ab} =0$, $\ell^2=-d(d-1)/R$, which is in agreement with the result (130) of \cite{Gurses:2014soa} obtained for Weyl type N AdS-plane waves  ($\tau_i = \Phi = 0$) in six dimensions.}}
In this case, $E_{ab}$ is traceless and there is no algebraic constraint determining the value of the Ricci scalar. Therefore, the vacuum field equations reduce to
\begin{equation}
\left( \mathcal{D} - \tfrac{R}{5} \right) \left( \mathcal{D} - \tfrac{2 R}{15} \right) \omega' = 0.
\label{ConfGrav}
\end{equation}
Similarly as in the case of quadratic gravity, for \tb{ TNS} metrics obtained by the Kerr-Schild transformation \eqref{KSUK}
of a \tb{Weyl type III} universal Kundt metric, the solution of \eqref{ConfGrav} \tb{(again $\omega'=-\mathcal{D}\mathcal{H}$)} can be found in the factorized form
$\mathcal{H} = \mathcal{H}_0 + \mathcal{H}_1 + \mathcal{H}_2$, where
\begin{equation}
\mathcal{D}\mathcal{H}_0 = 0, \quad
\left(\mathcal{D} - \tfrac{R}{5}\right)\mathcal{H}_1 = 0, \quad
\left(\mathcal{D} - \tfrac{2 R}{15}\right)\mathcal{H}_2 = 0.
\end{equation}
\tb{$\mathcal{H}_1$ and $\mathcal{H}_2$ generate non-Einstein solutions to six-dimensional conformal gravity.}

\subsubsection{Weyl type II solutions}

Let us consider the higher-dimensional generalization \eqref{GT-2blocks} of Khlebnikov-Ghanam-Thompson metrics
consisting of three 2-blocks for which
\begin{equation}
R\indices{_a^e_c^f} R_{bdef} S^{cd} =  \frac{R}{d} R_{acbd} S^{cd}, \quad R\indices{_a^e_c^f} R_{bfde} S^{cd} = \frac{R^2}{d^2} S_{ab}.
\end{equation}
Using these relations along with \eqref{KGT-SC}, \eqref{KGT-CC} and its covariant derivatives,
one can show that the tensor $E_{ab}$ of  conformal gravity \eqref{CCG} reduces to the form \eqref{TNS} with
\begin{equation} 
a_0 = - \frac{52\beta}{225} R^2, \quad
a_1 = \beta R, \quad
a_2 = -\beta
\end{equation}
and with the rest of the coefficients $\{a_i\}$ including $\lambda$ being zero.
Taking into account \eqref{boxy} and that $\mathcal{D} = \Box$, vacuum field equations of conformal gravity for KGT metrics reduce to
\begin{equation}
\left(\Box^2 - \frac{R}{3} \Box + \frac{2 R^2}{225}\right) \Box H = 0,
\label{ConfGrav2}
\end{equation}
with an arbitrary Ricci scalar.
The solution of the sixth order equation \eqref{ConfGrav2} can be written as a sum $H = H_0 + H_1 + H_2$
of solutions of the 2nd order equations
\begin{equation}
\Box H_0 = 0, \quad
\left(\Box -  \frac{5+ \sqrt{17}}{30}R \right) H_1 = 0, \quad
\left(\Box + \frac{5 + \sqrt{17}}{30}R \right) H_2 = 0.
\end{equation}

\begin{acknowledgments}
This work has been supported  by research plan RVO: 67985840 and by the Albert Einstein Center
for Gravitation and Astrophysics, Czech Science Foundation GACR 14-37086G.
\end{acknowledgments}

\appendix

\section{T-III spacetimes}
\label{app_TIII}

Here, we consider a slight generalization of TN spacetimes, so called T-III spacetimes, for which every symmetric rank-2 curvature tensor is of traceless type III:

\tb{
\begin{definition}[T-III spacetimes]
	\label{T-III}
T-III spacetimes 
are spacetimes, for which there exist a null vector $\bl$ and  $d-2$ spacelike vectors $m^{(i)a}$ such that for every symmetric rank-2 tensor $E_{ab}$ constructed polynomially from a metric, the Riemann tensor and  its covariant derivatives of an arbitrary order there exist a constant $\lambda$ and  functions $\phi$ and $\psi_i$ such that 
\be \label{TIII}
E_{ab} =\lambda g_{ab}+ \psi_i \ell_{(a} m^{(i)}_{ b)} + \phi \ell_a \ell_b.
\ee
\end{definition}
}

Thus, for T-III spacetimes, the vacuum field equations of any theory with the Lagrangian of the form \eqref{lagr} reduce 
to one algebraic equation and (at most) $d-1$ differential equations. Hence, also T-III spacetimes may be useful in finding solutions to higher-order gravity theories.    

It turns out that the necessary conditions for TN spacetimes naturally extend to the T-III class:

\begin{lemma}\label{lem_CSI_TIII}
	T-III spacetimes are CSI.
\end{lemma}
\begin{proof}
	Let us assume that a spacetime possesses a non-constant curvature invariant $I$ constructed polynomially from the Riemann tensor and its covariant derivatives of arbitrary order.
	$I$ can be expressed as a trace of a rank-2 tensor $E_{ab}$. Since
	the trace of $E_{ab}$ is non-constant then $E_{(ab)}$ is not of the form
	\eqref{TIII} 
	or more special and the spacetime is not T-III.
	Thus TN$\subset$T-III$\subset$CSI.
\end{proof}

\begin{proposition}[Necessary conditions for T-III spacetimes]
Non-Einstein {T-III} spacetimes are necessarily CSI Kundt spacetimes of Weyl type II or more special.
\end{proposition}
\begin{proof}
	Thanks to lemma \ref{lem_CSI_TIII}, T-III spacetimes are CSI.
	
	$\bl$ is geodetic: 
	The traceless Ricci tensors $S_{ab}$ is of type  III  \tb{(i.e., $S_{ab} = \psi'_i \ell_{(a} m^{(i)}_{ b)} + \omega' \ell_a \ell_b$)}
	and  b.w.\ $+2$ component of rank-2 tensor  $\nabla_a S_{cd} \nabla_b S^{cd}$ reads
	\be
	\nabla_a S_{cd} \nabla_b S^{cd} \ell^a \ell^b = 4 (\kappa_i  \psi'_i) (\kappa_j  \psi'_j)+2 \kappa_i \kappa_i   \psi'_i  \psi'_i\,.
	\ee
	For {T-III} 
	spacetimes, this b.w.\ $+2$ component has to vanish, which implies $\kappa_i=0$ and  thus $\bl$ is geodetic.

	$\bl$ is a  Kundt vector field: 
	For $S_{ab}$ of type III, the b.w.\ 0 components of $E_{ab} = \nabla_a S_{cd} \nabla_b S^{cd}$ have to satisfy
	\be
	E_{ab} m_{(i)}^a m_{(j)}^b = \delta_{ij} E_{ab} \ell^a n^b,
	\ee
	implying
	\be
	\psi'_k \psi'_k \rho_{li} \rho_{lj} + 2 \psi'_k \psi'_l \rho_{ki} \rho_{lj} = 0. \label{TTIIIKundt}
	\ee
	Contraction of $i$ and $j$ gives a sum of squares and therefore  $\rho_{ij}=0$ and the spacetime is Kundt.
	
	T-III spacetimes are of Weyl type II or more special:
	To prove this,  one can use the same arguments as in the proof of proposition \ref{nec_TN}.		
\end{proof}

Note that, to prove that T-III spacetimes are Kundt and algebraically special, it is sufficient to assume T-III$_{1}$.

Finally, for Kundt spacetimes of Weyl type III and traceless Ricci type III, all covariant derivatives  of the Riemann tensor $\nabla^{(k)} \BR$ are of aligned type III (proposition A.8 of \cite{universalMaxwell}) and thus

\begin{proposition}[Sufficient conditions for T-III spacetimes]
Weyl type N and III Kundt spacetimes with the Ricci tensor of the form \eqref{TIII} are T-III.
\end{proposition}

\subsection{T-III spacetimes in four and five dimensions}

First, let us study necessary conditions for algebraically special  T-III spacetimes.
In the four dimensional NP notation, the Weyl and Ricci tensors admit $\Psi_2$,  $\Psi_3$,  $\Psi_4$ and $\Phi_{12}={\bar \Phi}_{21}$, $\Phi_{22}$ components, respectively.

Similarly as for TN spacetimes,
for  {T-III}$_0$, $\Psi_2$ is constant. Bianchi equations (7.32a), (7.32b),
(7.32e) and (7.32h) from \cite{Stephanibook} then give again
\be
\kappa=\sigma=\rho=\tau=0 \, ,
\ee
respectively and thus we can immediately generalized lemma \ref{lem_4D_TNnec}
to T-III spacetimes

\begin{lemma}\label{lem_4D_TIIInec}
	In four dimensions, genuine Weyl type II and D    {T-III}$_0$ spacetimes are recurrent Kundt spacetimes.
\end{lemma}

The same arguments that have led to proposition \ref{prop_nonE5DTN} then imply
\begin{proposition}[Non-existence of 5D T-III Weyl type II spacetimes] \label{prop_nonE5DTIII}
In five dimensions, genuine Weyl type II and D {T-III}$_0$ spacetimes do not exist.	
\end{proposition}

\section{Kerr--Schild transformations of Einstein Kundt metrics}
\label{sec_KS}

A Kerr--Schild transformation of a metric $\bar g_{ab}$ is a transformation with
the transformed metric $g_{ab}$ being of the form
\begin{equation}
g_{ab} = \bar g_{ab} + 2 \H \ell_a \ell_b, \qquad
g^{ab} = (\bar g^{-1})^{ab} - 2 \H \ell^a \ell^b,
\end{equation}
where $\H$ is an arbitrary function and $\bl$ is a null vector.

Let us assume that the transformed metric $g_{ab}$ is Kundt with $\bl$ corresponding to the congruence of non-expanding,
non-shearing and non-twisting affinelly parametrized null geodesics.
Since
\begin{equation}
\kappa_i = \bar \kappa_i, \qquad
L_{10} = \bar L_{10}, \qquad
\rho_{ij} = \bar \rho_{ij},
\end{equation}
the vector $\bl$ has the same above mentioned geometrical properties in the background metric $\bar g_{ab}$ and thus
the background spacetime is necessarily Kundt as well.

Without loss of generality, we can always set the frame such that $L_{1i} = \tau_i$.
The frame components of the Ricci tensor then read
\begin{align}
\omega &= \bar \omega = 0, \qquad
\psi_i = \bar \psi_i, \qquad
\phi = \bar \phi + \D^2\H, \qquad
\phi_{ij} = \bar \phi_{ij}, \label{KS-omega}\\
\psi'_i &= \bar \psi'_i + \delta_i \D\H + \H \bar\psi_i, \\
\omega' &= \bar \omega' - \Xi_{ii} - \frac{2 \H}{n-1} \left( (n-3) \bar\phi - \bar\phi_{ii} \right) \label{KS-omegap}
\end{align}
and the independent components of the Weyl tensor are given by
\begin{align}
\Omega_{ij} &= \bar \Omega_{ij} = 0, \qquad
\Psi_{ijk} = \bar \Psi_{ijk} = \frac{2 \bar\psi_l}{n-2} \delta_{i[j} \delta_{k]l}, \qquad
\Phi^\text{A}_{ij} = \bar \Phi^\text{A}_{ij},\label{KS-Omega} \\
\Phi_{ijkl} &= \bar \Phi_{ijkl} + 4 \frac{\D^2\H}{(n-1)(n-2)} \delta_{i[k} \delta_{l]j}, \\
\Psi'_{ijk} &= \bar\Psi'_{ijk} + 2 \frac{\psi'_l - \bar\psi'_l}{n-2} \delta_{i[j} \delta_{k]l}, \\
\Omega'_{ij} &= \bar \Omega'_{ij} - \left( \Xi_{ij} - \frac{1}{n-2} \Xi_{kk} \delta_{ij} \right),\label{KS-Omegap}
\end{align}
where
\begin{equation}
\Xi_{ij} = \delta_{(i}\delta_{j)}\H + \rho'_{(ij)} \D\H + 2 \tau_{(i} \delta_{j)}\H + M^k_{(ij)} \delta_k\H
+ 2 \H \left( \tau_i \tau_j + \bar\Phi_{(ij)} + \frac{\bar\phi_{ij}}{n-2} \right).
\label{Xi}
\end{equation}
Note that $\D = \bar\D$, $\delta_i = \bar\delta_i$, and since the spacetime is Kundt, 
\begin{align}
\tau_i &= \bar \tau_i, \qquad
\rho'_{ij} = \bar \rho'_{ij}, \qquad
M^i_{jk} = \bar M^i_{jk}. \label{KS-tau}
\end{align}

Furthermore, for Einstein background spacetimes, i.e.\
\begin{equation}
\bar R_{ab} = \lambda \bar g_{ab},
\end{equation}
$S_{ab}$ is of type III iff $\D^2\H = 0$ and then the Ricci tensor of the full metric $g_{ab}$ reads
\begin{equation}
R_{ab} = \lambda g_{ab} + 2\delta_i \D\H \, \ell_{(a} m^{(i)}_{b)} + \xi \ell_a \ell_b,
\end{equation}
where
\begin{equation}
\xi = - \hat\nabla^2 \H - 2 \tau_i \delta_i\H - 2 \H \left( \tau_i \tau_i + \bar\Phi + \frac{n-2}{n-1} \lambda \right), \ \
\hat\nabla^2 = h^{ab} \nabla_a \nabla_b, \ \
h_{ab} = g_{ab} - 2 \ell_{(a} n_{b)}.
\end{equation}
Obviously, if in addition $\delta_i \D\H = 0$, $S_{ab}$ is of type N.
Possible combinations of the Weyl and trace-free Ricci types of  transformed metrics and background Einstein metrics
are summarized in Table \ref{table:KS}. Note that the most general Weyl type of Einstein Kundt spacetimes is type II.

\begin{table}[h]
	\centering
	\begin{tabular}{cc|cc}
		& & \multicolumn{2}{c}{types of $g_{ab}$} \\
		Weyl type of $\bar g_{ab}$ & conditions & Weyl & TF Ricci\\
		\hline
		II & ---  & II & III, N \\
		III & $\bar\Psi'_{ijk} + \frac{2}{n-2} \delta_{i[j} \delta_{k]} \D\H \neq 0$ & III & III, N \\
		III & $\bar\Psi'_{ijk} + \frac{2}{n-2} \delta_{i[j} \delta_{k]} \D\H = 0$ & N & III \\
		N & --- & III & III \\
		N & --- & N & N
	\end{tabular}
	\caption{Possible combinations of the Weyl and trace-free Ricci types of the Kerr--Schild transformed metric $g_{ab}$ with $\D^2{\cal H}=0$ depending on the Weyl type of the background Einstein metric $\bar{g}_{ab}$.}
	\label{table:KS}
\end{table}

\section{Variation of the FKWC basis up to order 6}
\label{FKWC}
To show what other terms than $\Box^n S_{ab}$ may appear in  field equations
for generic TN spacetimes, we provide a list of variations
 $\frac{1}{\sqrt{-g}} \frac{\delta (\sqrt{-g}I_k)}{\delta g^{ab}}$ of FKWC basis elements
\cite{FKWC} for Weyl type III TN spacetimes.

Any scalar curvature polynomial of order 6 (in derivatives of the metric) can be expressed in terms of 22 FKWC basis elements \cite{FKWC}. 
However, since variations of 7 of these are related to variations of the rest of the scalars via total divergence, only 15 of these 22 basis elements possess a nontrivial independent metric variation \cite{FKWCvariation, CSIuniversal}. 
Therefore, a general gravitational Lagrangian $\mathcal{L}$ of order 6 can be expanded in those 15 curvature invariants. 

For $CSI$ (and thus in particular for TN) spacetimes, variation of only 13 of these is non-vanishing and 4 of the remaining 13 scalars ($R^2$, $R^3$, $R R\indices{_{ab}}R\indices{^{ab}}$ and $R R\indices{_{abcd}}R\indices{^{abcd}}$) are functions of invariants of lower orders ($R$, $R\indices{_{ab}}R\indices{^{ab}}$ and $R\indices{_{abcd}}R\indices{^{abcd}}$) and hence their variation can be easily computed employing variations of their lower order counter-terms. 

Thus, given any gravitational Lagrangian of order 6, to compute the form of the field equations for TN spacetimes, 
metric variation of only 9 from the total of 22 FKWC invariants is needed. 
Moreover, if the TN metric is of Weyl type III, the form of these variations reduces dramatically:

\begin{align}
&R:
\quad -\frac{d-2}{d}Rg_{ab} + S_{ab},\\
&R\indices{_{ab}}R\indices{^{ab}}:
\quad -\frac{d-4}{2d^2} R^2 g_{ab} +\frac{2(d-2)}{d(d-1)}R S_{ab} +\square S_{ab} ,\\
&R\indices{_{abcd}}R\indices{^{abcd}}:
\quad -\frac{d-4}{d^2(d-1)} R^2 g_{ab} - \frac{4(d-2)}{d(d-1)}R S_{ab} + 4\square S_{ab} 
+2C\indices{_a^{cde}}C\indices{_{bcde}}, \\
&R\indices{_{ab}} \square R\indices{^{ab}}:
\quad -\frac{2R}{d(d-1)}\square S_{ab} + \square^2 S_{ab} ,\\
&R\indices{_{pq}}R\indices{^{p}_r}R\indices{^{qr}}:
\quad -\frac{d-6}{2d^3} R^3 g_{ab} +\frac{3(d-3)}{d^2(d-1)}R^2 S_{ab} + \frac{3R}{d} \square S_{ab}, \\
&R\indices{_{pq}}R\indices{_{rs}}R\indices{^{prqs}}:
\quad -\frac{d-6}{2d^3} R^3 g_{ab} + \frac{3d^2-10d +9}{d^2(d-1)^2}R^2 S_{ab} - \frac{3-2d}{d(d-1)}R \square S_{ab}, \\
&R\indices{_{pq}} R\indices{^{p}_{rst}}R\indices{^{qrst}}:
\quad -\frac{d-6}{d^3(d-1)}R^3 g_{ab} - \frac{2(2d^2-7d+9)}{d^2(d-1)^2} R^2 S_{ab} - \frac{4 R}{d-1}\square S_{ab}
+\frac{2 R}{d}C\indices{_{a}^{cde}}C\indices{_{bcde}} \nonumber \\
&\quad + C\indices{^{cdef}} \nabla_d \nabla\indices{_{(a}} C\indices{_{b)cef}}
+ C\indices{_{(a|cde|}} \square C\indices{_{b)}^{cde}} + \nabla\indices{_{(a}} C\indices{^{cfde}}\nabla\indices{_{|f|}}C\indices{_{b)cde}}
+ \nabla_f C\indices{_{bcde}}\nabla^f C\indices{_a^{cde}} ,\\
&R\indices{_{pqrs}} R\indices{^{pqab}}R\indices{^{rs}_{ab}}:
\quad -\frac{2(d-6)}{d^3(d-1)^2}R^3 g_{ab} + \frac{36R^2}{d^2(d-1)^2} S_{ab} - \frac{24 R^2}{d(d-1)^2}S_{ab} + \frac{24 R}{d(d-1)}\square S_{ab} \nonumber \\
&\quad + \frac{6(d+2)}{d(d-1)}RC\indices{_a^{cde}}C\indices{_{bcde}} +6\nabla_c C\indices{_{bfde}} \nabla^f C\indices{_a^{cde}} ,\\
&R\indices{_{prqs}} R\indices{^{p}_a^q_b}R\indices{^{rasb}}:
\quad-\frac{(d-6)(d-2)}{2d^3(d-1)^2}R^3 g_{ab} + 9 \frac{d-2}{d^2(d-1)^2} R^2 S_{ab}
- \frac{3 R}{d(d-1)}\square S_{ab} \nonumber \\
&\quad - \frac{3(d+2)}{2d(d-1)}R C\indices{_a^{cde}}C\indices{_{bcde}}
- 3\nabla_e C\indices{_{bdcf}} \nabla^f C\indices{_a^{cde}}
+ 3 C\indices{^{cdef}} \nabla_f \nabla_d C\indices{_{acbe}}.\label{FKWC-C9}
\end{align}

\section{Examples of type III TNS spacetimes and corresponding vacuum solutions to quadratic gravity}
\label{sec_examples}

\tb{
In this section, we employ  proposition \ref{prop_KS_TNS} to construct Weyl type III TNS spacetimes using the Kerr--Schild transformation \eqref{KSUK} and we find new vacuum solutions to quadratic gravity.}

\subsection{\tb{Ricci flat examples}}

\tb{For simplicity, as a background metric we consider 
recurrent ($\bar\tau_i = 0$
implies  $\tau_i = 0$, see Appendix \ref{sec_KS}) Ricci--flat  Weyl type III universal spacetimes.
Such spacetimes belong to VSI class and admit a metric of the form
\eqref{Kundt_deg} with metric functions  \cite{Coleyetal06} 
\begin{equation}
  H = H^{(0)}(u,x) + \frac{1}{2} (F - W_{m,m})r, \quad
  W_2 = 0, \quad W_m = W_m(u,x), \quad
  g_{\alpha\beta}(u,x) = \delta_{\alpha\beta}
  \label{VSI}
\end{equation}
that satisfy the Ricci-flat condition
\begin{equation}
  \Delta H^{(0)} - \tfrac{1}{4} W_{mn} W_{mn} - 2 H^{(1)}_{,m} W_m - H^{(1)} W_{m,m} - W_{m,mu} = 0
  \label{VSI_Ricciflat}
\end{equation}
and $F = F(u,x^i)$ is subject to
\begin{equation}
  F_{,2} = 0, \quad F_{,m} = \Delta W_m,
  \label{VSI_F}
\end{equation}
where $W_{mn} = W_{m,n} - W_{n,m}$ and $\Delta \equiv \partial^i\partial_i$ is the spacial Laplacian.
The natural null frame of VSI metrics
\begin{equation}
  \bar\ell_a \dd x^a = \dd u, \quad
  \bar n_a \dd x^a = \dd r + H \dd u + W_\alpha \dd x^\alpha, \quad
  \bar m_a^{(\alpha)} \dd x^a = \dd x^\alpha,
\end{equation}
is parallelly propagated in the recurrent case and indeed $\bar\tau_i = \bar L_{1i} = 0$.}
\tb{For  geodetic $\bl$  and $\D f = 0$ (i.e., $\partial_r f = 0$),
	the action of the operator $\mathcal{D}$ introduced in \eqref{calD} on $f$ reduces to
	\begin{equation}
	\mathcal{D} f = \Delta f.
	\end{equation}
}

\tb{The rank-2 tensor $F_2$ vanishes identically for the background metric since $\bar\tau_i = 0$
and it remains to satisfy $F_0 = 0$, i.e.,
\begin{equation}
  \bar\Psi'_{ijk} \bar\Psi'_{ijk} = 2 \bar\Psi'_i \bar\Psi'_i,
  \label{F0_VSI}
\end{equation}
where \tb{ $\bar\Psi'_i = -H^{(1)}_{,i}$} and $\bar\Psi'_{ijk} = \frac{1}{2}W_{kj,i}$.
Note that the trivial solution with both sides of \eqref{F0_VSI} vanishing  corresponds to Weyl type N.}

\tb{In four dimensions,  $F_0$  vanishes identically and thus in this case, all Kundt metrics obeying \eqref{VSI} and \eqref{VSI_F} are TNS (c.f. proposition \ref{prop_TNS}). In particular, note that the Einstein equation \eqref{VSI_Ricciflat} is not a necessary condition for TNS. However, it is  useful to start with  a background Einstein spacetime and express $H^{(0)}$ as $H^{(0)}=H^{(0)}_{E}+ \mathcal{H}$, where $H^{(0)}_{E}$ corresponds to a background Einstein solution of \eqref{VSI_Ricciflat} and $\mathcal{H}$ represents the Kerr--Schild transformation. Then the 
field equation of quadratic gravity \eqref{QGeq} reduces to
\begin{equation}
  \left( \Delta  + \tfrac{1}{\kappa \beta}  \right) \Delta  \mathcal{H} = 0.
\end{equation}
 Clearly, the Kerr--Schild transformations with $\Delta  \mathcal{H} = 0$ takes the metric from an Einstein spacetime to an Einstein spacetime (c.f. eq. \eqref{VSI_Ricciflat}). The non-Einstein vacuum solutions to quadratic gravity are obtained by  solving }
\begin{equation}
\tb{\left( \Delta  + \tfrac{1}{\kappa \beta}  \right)  \mathcal{H} = 0.}
\label{QGeq_pureQG}
\end{equation}
\tb{ An example of such a solution is }
\be
\tb{	\mathcal{H}=\left({c_1}(u) e^{\sqrt{k_1}x}+{c_2}(u) e^{-\sqrt{{k_1}}x}\right)\left({c_3}(u) e^{\sqrt{k_2}
	y}+{c_4}(u) e^{-\sqrt{{k_2}}y}\right),\ \  k_1+k_2+\tfrac{1}{\kappa\beta}=0.}
\ee
\tb{To conclude, by adding a solution of \eqref{QGeq_pureQG},
 $\mathcal{H}$, to any solution of \eqref{VSI_Ricciflat}, $H^{(0)}_{E}$, we obtain a vacuum solution to quadratic gravity, $H^{(0)}=H^{(0)}_{E}+ \mathcal{H}$.}

\tb{In contrast with the four-dimensional case, in higher dimensions the equation \eqref{F0_VSI} is non-trivial.
An explicit  solution of \eqref{VSI_F} and \eqref{F0_VSI} is, for instance,
\begin{equation}
  W_3 = (\alpha(u) \sin x_3 + \beta(u) \cos x_3) e^{x_4}, \quad
  W_4 = (\gamma(u) \sin x_3 + \delta(u) \cos x_3) e^{x_4}, \quad
  F = F(u), \label{HD-VSI}
\end{equation}
where $\alpha$, $\beta$, $\gamma$, $\delta$, $F$ are arbitrary functions of $u$
and all other $W_m$ vanish.
If moreover $\delta = -\alpha$ and $\gamma = \beta$, then the Ricci-flat condition \eqref{VSI_Ricciflat}
simplifies to the Laplace equation $\Delta H^{(0)} = 0$ for $H^{(0)}$.}

\tb{The metric \eqref{Kundt_deg}, satisfying \eqref{VSI}, \eqref{VSI_F} and \eqref{HD-VSI} is thus TNS in any dimension. To find a solution to quadratic gravity one can proceed similarly as in four dimensions.
}

\subsection{\tb{Examples with a non-vanishing Ricci scalar}}
\label{sec_ex_Ricci}

\tb{So far we have studied cases with vanishing Ricci scalar. In this section,
using propositions \ref{prop_warp_U} and \ref{prop_KS_TNS},  we construct
examples of Weyl type III TNS metrics with negative Ricci scalar. As a background metric $\bg_\text{UK}$ we use  a warp product of any of the above-mentioned Ricci-flat VSI metrics satisfying $F_0=0$
and one extra flat dimension \eqref{warp_productU}. 
The Ricci scalar of $\bg_\text{UK}$ 
is $\bar R = d(d-1)\lambda$.}
The d'Alembert operators of the warped and seed metric are related by
\begin{equation}
  \bar\Box f = -\lambda z^2 (\tilde\Box f + \partial_z\partial_z f) + (d-2) \lambda z \partial_z f.
\end{equation}
Choosing a parallelly propagated frame 
\begin{align}
  &\bar\ell_a \dd x^a = \dd u, \quad
  \bar n_a \dd x^a = \tfrac{1}{-\lambda z^2}\left(\dd r + \left(H - \tfrac{r^2}{2 z^2}\right) \dd u + W_i \dd x^i - \tfrac{r}{z} \dd z\right), \nonumber\\
  &\bar m_a^{(i)} \dd x^a = \tfrac{1}{\sqrt{-\lambda z^2}} \dd x^i, \quad
  \bar m_a^{(z)} \dd x^a = \tfrac{1}{\sqrt{-\lambda z^2}} \left(\dd z + \tfrac{r}{z} \dd u\right),
\end{align}
a straightforward calculation shows that $\bar\tau_i \bar\tau_i = -\lambda$, $\bar L_{1i} = \bar\tau_i$
and $\bar\tau_i \bar\delta_i = -\lambda(z \bar\nabla_z + r \bar\nabla_r)$. \tb{
Hence, the Kerr--Schild transformation \eqref{KSUK} with $\ell_a \dd x^a = \dd u$ and
$\bg_\text{UK}$ being a warp product \eqref{warp_productU}
of a universal Weyl type III VSI metric \eqref{Kundt_deg}, thus obeying \eqref{VSI}--\eqref{VSI_F}, and \eqref{F0_VSI},
is a Weyl type III TNS metric with $R = d(d-1)\lambda$.
 For $f$, $\partial_r f = 0$, the operator $\mathcal{D}$ defined in \eqref{calD}  reduces to
}
\begin{equation}
  \mathcal{D}f = -\lambda z^2 \Delta f + (d-6) \lambda z \partial_z f + 2(d-3)\lambda f,
\end{equation}
where $\Delta = \tilde\Delta + \partial_z\partial_z$.

\end{document}